\let\accentvec\vec  % to remove the warning about vec
\documentclass{llncs}
\let\vec\accentvec %ditto

\usepackage{xspace}
\usepackage{amssymb}
\usepackage[shortlabels]{enumitem}
\setlist[enumerate,1]{label=\arabic*., ref=(\arabic*)}

\usepackage{url}
\usepackage{amsmath}

\usepackage[pdftex,                %%% hyper-references for pdflatex
            bookmarks=true,        %%% generate bookmarks ...
            bookmarksnumbered=true,%%% ... with numbers
            bookmarksopen=false,
            hypertexnames=false,
]{hyperref}
\makeatletter
\providecommand*{\toclevel@title}{0}
\providecommand*{\toclevel@author}{0}
\makeatother

\usepackage{xcolor}
\definecolor{dark-red}{rgb}{0.4,0.15,0.15}
\definecolor{dark-blue}{rgb}{0.15,0.15,0.4}
\definecolor{medium-blue}{rgb}{0,0,0.5}
\definecolor{gray}{rgb}{0.5,0.5,0.5}

\hypersetup{
    colorlinks, linkcolor={dark-red},
    citecolor={dark-blue}, urlcolor={medium-blue}
}

\usepackage{doi}

\newcommand{\problemdef}[3]
{
\begin{quote}
\textsc{#1}\\
\textbf{Input:} #2\\
\textbf{Question:} #3
\end{quote}
}

\newcommand{\parproblemdef}[4]
{
\begin{quote}
\textsc{#1}\\
\textbf{Input:} #2\\
\textbf{Parameter:} #3\\
\textbf{Question:} #4
\end{quote}
}

\newcommand{\Cutwidth}[0]{\textsc{Cutwidth}\xspace}
\newcommand{\cutwidththree}[0]{\textsc{Cutwidth3}\xspace}
\newcommand{\N}{\ensuremath{\mathbb{N}}\xspace}

\newcommand{\yes}[0]{\textsc{yes}\xspace}
\newcommand{\no}[0]{\textsc{no}\xspace}

\newcommand{\OR}[0]{\ensuremath{\mathop{\mathrm{\textsc{or}}}}\xspace}
\newcommand{\AND}[0]{\ensuremath{\mathop{\mathrm{\textsc{and}}}}\xspace}

\newcommand{\qCNFSAT}[0]{\ensuremath{\mathop{\mathrm{\textsc{$q$-CNF-SAT}}}}\xspace}

\newcommand{\containment}[0]{NP~$\subseteq$ coNP$/$poly\xspace}
\newcommand{\ncontainment}[0]{NP~$\not \subseteq$ coNP$/$poly\xspace}

\newcommand{\Q}[0]{\ensuremath{\mathcal{Q}}\xspace}

\newcommand{\eqvr}[0]{\ensuremath{\mathcal{R}}\xspace}

\let\plainsquareforqed\squareforqed
\newcommand{\claimqed}{\renewcommand{\squareforqed}{$\diamondsuit$}\qed\renewcommand{\squareforqed}{\plainsquareforqed}}

\newcommand{\vc}{\ensuremath{\mathrm{\textsc{vc}}}\xspace}
\newcommand{\fvs}{\ensuremath{\mathrm{\textsc{fvs}}}\xspace}
\newcommand{\tw}{\ensuremath{\mathrm{\textsc{tw}}}\xspace}

\newcommand{\Oh}[0]{\ensuremath{\mathcal{O}}\xspace}

\newlength{\baseImageHeight}
\newcommand{\hyphen}{\nobreakdash-\hspace{0pt}}

\newcommand{\Pathwidth}[0]{\textsc{Pathwidth}\xspace}

\newcommand{\kVertexCover}[0]{\textsc{$k$\hyphen Vertex Cover}\xspace}
\newcommand{\nTreewidth}[0]{\textsc{$n$\hyphen Treewidth}\xspace}
\newcommand{\nPathwidth}[0]{\textsc{$n$\hyphen Pathwidth}\xspace}
\newcommand{\TreewidthVC}[0]{\textsc{Treewidth [vc]}\xspace}
\newcommand{\PathwidthVC}[0]{\textsc{Pathwidth [vc]}\xspace}

\newcommand{\ElimCobipartite}[0]{\textsc{Cobipartite Graph Elimination}\xspace}
\newcommand{\Satisfiability}[0]{\textsc{Satisfiability}\xspace}

\newcommand{\Treewidth}[0]{\textsc{Tree\-width}\xspace}

\newtheorem{observation}{Observation}
\newtheorem{numberedclaim}{Claim}

\spnewtheorem*{qexpansion}{$q$-Expansion Lemma}{\bfseries}{\itshape}

\newcommand{\sectref}[1]{Section~\ref{#1}}

\newcommand{\defref}[1]{Definition~\ref{#1}}
\newcommand{\claimref}[1]{Claim~\ref{#1}}
\newcommand{\lemmaref}[1]{Lemma~\ref{#1}}

\newcommand{\thmref}[1]{Theorem~\ref{#1}}

\newcommand{\obsref}[1]{Observation~\ref{#1}}

\title{On Sparsification for Computing Treewidth\thanks{This work was supported by ERC Starting Grant 306992 ``Parameterized Approximation''.}}

\author{Bart M.\ P.\ Jansen \inst{1}}

\institute{
University of Bergen, Norway.
\email{bart.jansen@ii.uib.no}
}

\begin{document}

% Suppress bookmarks for title and author.
\hypersetup{bookmarksdepth=-1}

\maketitle

\hypersetup{bookmarksdepth=2}

\begin{abstract}
We investigate whether an $n$-vertex instance~$(G,k)$ of \Treewidth, asking whether the graph~$G$ has treewidth at most~$k$, can efficiently be made sparse without changing its answer. By giving a special form of \OR-cross-composition, we prove that this is unlikely: if there is an~$\epsilon > 0$ and a polynomial-time algorithm that reduces $n$-vertex \Treewidth instances to equivalent instances, of an arbitrary problem, with~$\Oh(n^{2 - \epsilon})$ bits, then \containment and the polynomial hierarchy collapses to its third level.

Our sparsification lower bound has implications for structural parameterizations of \Treewidth: parameterizations by measures that do not exceed the vertex count, cannot have kernels with~$\Oh(k^{2 - \epsilon})$ bits for any~$\epsilon > 0$, unless \containment. Motivated by the question of determining the optimal kernel size for \Treewidth parameterized by vertex cover, we improve the~$\Oh(k^3)$-vertex kernel from Bodlaender et al.~(STACS 2011) to a kernel with~$\Oh(k^2)$ vertices. Our improved kernel is based on a novel form of \emph{treewidth-invariant set}. We use the $q$-expansion lemma of Fomin et al.~(STACS 2011) to find such sets efficiently in graphs whose vertex count is superquadratic in their vertex cover number.
\end{abstract}

\section{Introduction}

The task of preprocessing inputs to computational problems to make them less dense, called \emph{sparsification}, has been studied intensively due to its theoretical and practical importance. Sparsification, and more generally, preprocessing, is a vital step in speeding up resource-demanding computations in practical settings. In the context of theoretical analysis, the \emph{sparsification lemma} due to Impagliazzo et al.~\cite{ImpagliazzoPZ01} has proven to be an important asset for studying subexponential-time algorithms. The work of Dell and van Melkebeek~\cite{DellM10} on sparsification for \Satisfiability has led to important advances in the area of kernelization lower bounds. They proved that for all~$\epsilon > 0$ and~$q \geq 3$, assuming \ncontainment, there is no polynomial-time algorithm that maps an instance of \qCNFSAT on~$n$ variables to an equivalent instance on~$\Oh(n^{q - \epsilon})$ bits --- not even if it is an instance of a \emph{different} problem. 
%Hence the trivial sparsification scheme that encodes a \qCNFSAT formula by~$(2n)^q$ bits, indicating for each of the~$(2n)^q$ clauses whether or not they are present in the instance, appears to be asymptotically optimal.

This paper deals with sparsification for the task of building minimum-width tree decompositions of graphs, or, in the setting of decision problems, of determining whether the treewidth of a graph~$G$ is bounded by a given integer~$k$. Preprocessing procedures for \Treewidth have been studied in applied~\cite{BodlaenderK06c,BodlaenderKE05,EijkhofBK07} and theoretical settings~\cite{BodlaenderDFH09,BodlaenderJK11b}. A team including the current author obtained~\cite{BodlaenderJK11b} a polynomial-time algorithm that takes an instance~$(G,k)$ of \Treewidth, and produces in polynomial time a graph~$G'$ such that~$\tw(G) \leq k$ if and only if~$\tw(G') \leq k$, with the guarantee that~$|V(G')| \in \Oh(\vc^3)$ ($\vc$ denotes the size of a smallest vertex cover of the input graph). A similar algorithm was given that reduces the vertex count of~$G'$ to~$\Oh(\fvs^4)$, where~$\fvs$ is the size of a smallest feedback vertex set in~$G$. Hence polynomial-time data reduction can compress \Treewidth instances to a number of 
vertices polynomial in their vertex cover (respectively feedback vertex) number. On the other hand, the natural parameterization of \Treewidth is trivially \AND-compositional, and therefore does not admit a polynomial kernel unless \containment~\cite{BodlaenderDFH09,Drucker12}. These results give an indication of how far the vertex count of a \Treewidth instance can efficiently be reduced in terms of various measures of its complexity. However, they do not tell us anything about the question of \emph{sparsification}: can we efficiently make a \Treewidth instance less dense, without changing its answer?

\textbf{Our results.} 
Our first goal in this paper is to determine whether nontrivial sparsification is possible for \Treewidth instances. As a simple graph~$G$ on~$n$ vertices can be encoded in~$n^2$ bits through its adjacency matrix, \Treewidth instances consisting of a graph~$G$ and integer~$k$ in the range~$[1 \ldots n]$ can be encoded in~$\Oh(n^2)$ bits. We prove that it is unlikely that this trivial sparsification scheme for \Treewidth can be improved significantly: if there is a polynomial-time algorithm that reduces \Treewidth instances on~$n$ vertices to equivalent instances of an arbitrary problem, with~$\Oh(n^{2 - \epsilon})$ bits, for some~$\epsilon > 0$, then \containment and the polynomial hierarchy collapses~\cite{Yap83}. We prove this result by giving a particularly efficient form of \OR-cross-composition~\cite{BodlaenderJK12f}. We embed the \OR of~$t$ $n$-vertex instances of an NP-complete graph problem into a \Treewidth instance with~$\Oh(n \sqrt{t})$ vertices. The construction is a combination of three 
ingredients. 
We carefully inspect the properties of Arnborg et al.'s~\cite{ArnborgCP87} NP-completeness proof for \Treewidth to obtain an NP-complete source problem called \ElimCobipartite that is amenable to composition. Its instances have a restricted form that ensures that good solutions to the composed \Treewidth instance cannot be obtained by combining partial solutions to two different inputs. Then, like Dell and Marx~\cite{DellM12}, we use the layout of a~$2 \times \sqrt{t}$ table to embed~$t$ instances into a graph on~$\Oh(n^{\Oh(1)}\sqrt{t})$ vertices. For each way of choosing a cell in the top and bottom row, we embed one instance into the edge set induced by the vertices representing the two cells. Finally, we use ideas employed by Bodlaender et al.~\cite{BodlaenderJK12a} in the superpolynomial lower bound for \Treewidth parameterized by the vertex-deletion distance to a clique: we compose the input instances of \ElimCobipartite into a cobipartite graph to let the resulting \Treewidth instance express a 
logical \OR, rather than an \AND. Our proof combines these three ingredients with an intricate analysis of the behavior of elimination orders on the constructed instance. As the treewidth of the constructed cobipartite graph equals its pathwidth~\cite{Mohring96}, the obtained sparsification lower bound for \Treewidth also applies to \Pathwidth.

Our sparsification lower bound has immediate consequences for parameterizations of \Treewidth by graph parameters that do not exceed the vertex count, such as the vertex cover number or the feedback vertex number. Our result shows the impossibility of obtaining kernels of bitsize~$\Oh(k^{2 - \epsilon})$ for such parameterized problems, assuming \ncontainment. The kernel for \Treewidth parameterized by vertex cover (\TreewidthVC) obtained by Bodlaender et al.~\cite{BodlaenderJK11} contains~$\Oh(\vc^3)$ vertices, and therefore has bitsize~$\Omega(\vc^4)$. Motivated by the impossibility of obtaining kernels with~$\Oh(\vc^{2 - \epsilon})$ bits, and with the aim of developing new reduction rules that are useful in practice, we further investigate kernelization for \TreewidthVC. We give an improved kernel based on \emph{treewidth-invariant sets}: independent sets of vertices whose elimination from the graph has a predictable effect on its treewidth. While finding such sets seems to be hard in general, we show that 
the $q$-expansion lemma, previously employed by Thomass\'e~\cite{Thomasse10} and Fomin et al.~\cite{FominLMPS11}, can be used to find them when the graph is large with respect to its vertex cover number. The resulting kernel shrinks \Treewidth instances to~$\Oh(\vc^2)$ vertices, allowing them to be encoded in~$\Oh(\vc^3)$ bits. Thus we reduce the gap between the upper and lower bounds on kernel sizes for \TreewidthVC. Our new reduction rule for \TreewidthVC relates to the old rules like the crown-rule for \kVertexCover relates to the high-degree Buss-rule~\cite{BussG93}: by exploiting local optimality considerations, our reduction rule does not need to know the value of~$k$.

\textbf{Related work.} While there is an abundance of superpolynomial kernel lower bounds, few superlinear lower bounds are known for problems admitting polynomial kernels. There are results for hitting set problems~\cite{DellM10}, packing problems~\cite{DellM12,HermelinW12}, and for domination problems on degenerate graphs~\cite{CyganGH13}.

\section{Preliminaries} \label{section:preliminaries}
\textbf{Parameterized complexity and kernels.}
A parameterized problem~$\Q$ is a subset of~$\Sigma^* \times \mathbb{N}$. The second component of a tuple~$(x,k) \in \Sigma^* \times \mathbb{N}$ is called the \emph{parameter}~\cite{DowneyF99,FlumG06}. The set~$\{1, 2, \ldots, n\}$ is abbreviated as~$[n]$. For a finite set~$X$ and integer~$i$ we use~$\binom{X}{i}$ to denote the collection of size-$i$ subsets of~$X$.

\begin{definition}[Generalized kernelization]\label{definition:kernelization}
Let~$\Q,\Q'\subseteq\Sigma^*\times\N$ be parameterized problems and let~$h\colon\N\to\N$ be a computable function. A \emph{generalized kernelization for~$\Q$ into~$\Q'$ of size~$h(k)$} is an algorithm that, on input~$(x,k)\in\Sigma^*\times\N$, takes time polynomial in~$|x|+k$ and outputs an instance~$(x',k')$ such that:
\begin{quote}
\begin{itemize}
	\item $|x'|$ and~$k'$ are bounded by~$h(k)$.
	\item $(x',k') \in \Q'$ if and only if~$(x,k) \in \Q$.
\end{itemize}
\end{quote}
The algorithm is a \emph{kernelization}, or in short a \emph{kernel}, for~$\Q$ if~$\Q'=\Q$. It is a \emph{polynomial (generalized) kernelization} if~$h(k)$ is a polynomial.
\end{definition}

\textbf{Cross-composition.} To prove our sparsification lower bound, we use a variant of cross-composition tailored towards lower bounds on the degree of the polynomial in a kernel size bound. The extension is discussed in the journal version~\cite{BodlaenderJK12f} of the extended abstract on cross-composition~\cite{BodlaenderJK11}.

\begin{definition}[Polynomial equivalence relation] \label{polyEquivalenceRelation}
An equivalence relation~\eqvr on~$\Sigma^*$ is called a \emph{polynomial equivalence relation} if the following conditions hold:
\begin{enumerate}
	\item There is an algorithm that given two strings~$x,y \in \Sigma^*$ decides whether~$x$ and~$y$ belong to the same equivalence class in time polynomial in~$|x| + |y|$.
	\item For any finite set~$S \subseteq \Sigma^*$ the equivalence relation~$\eqvr$ partitions the elements of~$S$ into a number of classes that is polynomially bounded in the size of the largest element of~$S$.
\end{enumerate}
\end{definition}  

\begin{definition}[Cross-composition]\label{definition:crosscomposition}
Let~$L\subseteq\Sigma^*$ be a language, let~$\eqvr$ be a polynomial equivalence relation on~$\Sigma^*$, let~$\Q\subseteq\Sigma^*\times\N$ be a parameterized problem, and let~$f \colon \N \to \N$ be a function. An \emph{\OR-cross-com\-position of~$L$ into~$\Q$} (with respect to \eqvr) \emph{of cost~$f(t)$} is an algorithm that, given~$t$ instances~$x_1, x_2, \ldots, x_t \in \Sigma^*$ of~$L$ belonging to the same equivalence class of~$\eqvr$, takes time polynomial in~$\sum _{i=1}^t |x_i|$ and outputs an instance~$(y,k) \in \Sigma^* \times \mathbb{N}$ such that:
\begin{quote}
\begin{itemize}
\item The parameter~$k$ is bounded by $\Oh(f(t)\cdot(\max_i|x_i|)^c)$, where~$c$ is some constant independent of~$t$.
\item $(y,k) \in \Q$ if and only if there is an~$i \in [t]$ such that~$x_i \in L$.\label{property:OR}
\end{itemize}
\end{quote}
\end{definition}

\begin{theorem}[{\cite[Theorem 6]{BodlaenderJK12f}}] \label{theorem:orcrosscomposition:polynomialbound}
Let~$L\subseteq\Sigma^*$ be a language, let~$\Q\subseteq\Sigma^*\times\N$ be a parameterized problem, and let~$d,\epsilon$ be positive reals. If~$L$ is NP-hard under Karp reductions, has an \OR-cross-composition into~$\Q$ with cost~$f(t)=t^{1/d+o(1)}$, where~$t$ denotes the number of instances, and~$\Q$ has a polynomial (generalized) kernelization with size bound~$\Oh(k^{d-\epsilon})$, then \containment.
\end{theorem}

\textbf{Graphs.}
All graphs we consider are finite, simple, and undirected. An undirected graph~$G$ consists of a vertex set~$V(G)$ and an edge set~$E(G) \subseteq \binom{V(G)}{2}$. The open neighborhood of a vertex~$v$ in graph~$G$ is denoted~$N_G(v)$, while its closed neighborhood is~$N_G[v]$. The open neighborhood of a set~$S \subseteq V(G)$ is~$N_G(S) := \bigcup _{v \in S} N_G(v) \setminus S$, while the closed neighborhood is~$N_G[S] := N_G(S) \cup S$. If~$S \subseteq V(G)$ then~$G[S]$ denotes the subgraph of~$G$ induced by~$S$. We use~$G - S$ to denote the graph~$G[V(G) \setminus S]$ that results after deleting all vertices of~$S$ and their incident edges from~$G$. A graph is \emph{cobipartite} if its edge-complement is bipartite. Equivalently, a graph~$G$ is cobipartite if its vertex set can be partitioned into two sets~$X$ and~$Y$, such that both~$G[X]$ and~$G[Y]$ are cliques. A matching~$M$ in a graph~$G$ is a set of edges whose endpoints are all distinct. The endpoints of the edges in~$M$ are \emph{saturated} by the 
matching. For disjoint subsets~$A$ and~$B$ of a graph~$G$, we say that~$A$ has a perfect matching into~$B$ if there is a matching that saturates~$A \cup B$ such that each edge in the matching has exactly one endpoint in each set. If~$\{u,v\}$ is an edge in graph~$G$, then \emph{contracting~$\{u,v\}$ into~$u$} is the operation of adding edges between~$u$ and~$N_G(v)$ while removing~$v$. A graph~$H$ is a \emph{minor} of a graph~$G$, if~$H$ can be obtained from a subgraph of~$G$ by edge contractions.

\textbf{Treewidth and Elimination Orders.}
While treewidth~\cite{Bodlaender98} is commonly defined in terms of tree decompositions, for our purposes it is more convenient to work with an alternative characterization in terms of \emph{elimination orders}. \emph{Eliminating} a vertex~$v$ in a graph~$G$ is the operation of removing~$v$ while completing its open neighborhood into a clique, i.e., adding all missing edges between neighbors of~$v$. An elimination order of an $n$-vertex graph~$G$ is a permutation~$\pi \colon V(G) \to [n]$ of its vertices. Given an elimination order~$\pi$ of~$G$, we obtain a series of graphs by consecutively eliminating~$\pi^{-1}(1), \ldots, \pi^{-1}(n)$ from~$G$. The \emph{cost} of eliminating a vertex~$v$ according to the order~$\pi$, is the size of the \emph{closed neighborhood} of~$v$ at the moment it is eliminated. The \emph{cost of~$\pi$ on~$G$}, denoted~$c_G(\pi)$, is defined as the maximum cost over all vertices of~$G$.

\begin{theorem}[{\cite[Theorem 36]{Bodlaender98}}] \label{theorem:tw:eliminationcharacterization}
The treewidth of a graph~$G$ is exactly one less than the minimum cost of an elimination order for~$G$.
\end{theorem}

\begin{lemma}[{\cite[Lemma 4]{BodlaenderFKKT06}, cf.~\cite[Lemma 6.13]{Jansen13}}] \label{lemma:eliminateOnePartiteSet}
Let~$G$ be a graph containing a clique~$B \subseteq V(G)$, and let~$A := V(G) \setminus B$. There is a minimum-cost elimination order~$\pi^*$ of~$G$ that eliminates all vertices of~$A$ before eliminating any vertex of~$B$.
\end{lemma}

\noindent Following the notation employed by Arnborg et al.~\cite{ArnborgCP87} in their NP-completeness proof, we say that a \emph{block} in a graph~$G$ is a maximal set of vertices with the same closed neighborhood. An elimination order~$\pi$ for~$G$ is \emph{block-contiguous} if for each block~$S \subseteq V(G)$, it eliminates the vertices of~$S$ contiguously. The following observation implies that every graph has a block-contiguous minimum-cost elimination order.

\begin{observation} \label{observation:finishBlock}
Let~$G$ be a graph containing two adjacent vertices~$u,v$ such that~$N_G[u] \subseteq N_G[v]$. Let~$\pi$ be an elimination order of~$G$ that eliminates~$v$ before~$u$, and let the order~$\pi'$ be obtained by updating~$\pi$ such that it eliminates~$u$ just before~$v$. Then the cost of~$\pi'$ is not higher than the cost of~$\pi$.
\end{observation}

\section{Sparsification Lower Bound for Treewidth} \label{section:lowerbound}
In this section we give the sparsification lower bound for \Treewidth. We phrase it in terms of a kernelization lower bound for the parameterization by the number of vertices, formally defined as follows.

\parproblemdef{\nTreewidth}
{An integer~$n$, an $n$-vertex graph~$G$, and an integer~$k$.}
{The number of vertices~$n$.}
{Is the treewidth of~$G$ at most~$k$?}

\noindent The remainder of this section is devoted to the proof of the following theorem.

\begin{theorem} \label{theorem:treewidth:sparsification}
If \nTreewidth admits a (generalized) kernel of size $\Oh(n^{2 - \epsilon})$, for some~$\epsilon > 0$, then \containment.
\end{theorem}

We prove the theorem by cross-composition. We therefore first define a suitable source problem for the composition in \sectref{section:sourceProblem}, give the construction of the composed instance in \sectref{section:construction}, analyze its properties in \sectref{section:analyzeComposedInstance}, and finally put it all together in \sectref{section:sparsificationProof}.

\subsection{The Source Problem} \label{section:sourceProblem}

The sparsification lower bound for \Treewidth will be established by cross-composing the following problem into it.

\problemdef{\ElimCobipartite}
{A cobipartite graph~$G$ with partite sets~$A$ and~$B$, and a positive integer~$k$, such that the following holds:~$|A|=|B|$,~$|A|$ is even,~$k < \frac{|A|}{2}$, and~$A$ has a perfect matching into~$B$.}
{Is there an elimination order for~$G$ of cost at most~$|A| + k$?}

\noindent The NP-completeness proof extends the completeness proof for \Treewidth~\cite{ArnborgCP87}.

\begin{lemma} \label{lemma:elimCobipartiteNPC}
\ElimCobipartite is NP-complete.
\end{lemma}
\begin{proof}
Membership in NP is trivial. To establish completeness, we use the connection between treewidth and elimination orders. The instances created by the NP-completeness proof for \Treewidth due to Arnborg et al.~\cite{ArnborgCP87} are close to satisfying the desired conditions. In Section 3 of their paper, Arnborg et al.~\cite{ArnborgCP87} reduce the \Cutwidth problem to \Treewidth. They show how to transform an $n$-vertex graph~$G$ with maximum degree~$\Delta$ into a cobipartite graph~$G'$ with partite sets~$A$ and~$B$ of size~$(\Delta + 1)n$, such that~$G$ has cutwidth at most~$k$ if and only if~$G'$ has treewidth at most~$(\Delta+1)(n+1) + k - 1 = |A| + \Delta + k$. By \thmref{theorem:tw:eliminationcharacterization} the latter happens if and only if~$G'$ has an elimination order of cost at most~$|A| + \Delta + k + 1$. It is easy to verify that their construction results in a graph with a perfect matching between the sets~$A$ and~$B$.

Using this information we prove the NP-completeness of \ElimCobipartite. We reduce from \cutwidththree (cf.~\cite[\S 5]{BodlaenderJK12a}), the cutwidth problem on subcubic graphs, which is known to be NP-complete~\cite[Corollary 2.10]{MonienS88}. Given an instance~$(G,k)$ of \cutwidththree, let~$n$ be the number of vertices in~$G$. As the cutwidth of a graph does not exceed its edge count, and a subcubic $n$-vertex graph has at most~$3n$ edges, we may output a constant-size \yes-instance if~$k \geq 3n$. In the remainder we therefore have~$k < 3n$. Form a new graph~$G^*$ as the disjoint union of~$20$ copies of~$G$. The resulting graph~$G^*$ has~$20n$ vertices, and its maximum degree is~$\Delta \leq 3$. As the cutwidth of a graph is the maximum cutwidth of its connected components, graph~$G^*$ has cutwidth at most~$k$ if and only if~$(G,k)$ is a \yes-instance. Now apply the transformation by Arnborg et al.\ to the instance~$(G^*,k)$. It results in a cobipartite graph~$G'$ with partite sets~$A'$ and~$B'$ of 
size~$(\Delta+1)20n$, such that~$G'$ has an elimination order of cost~$|A'| + \Delta + k + 1$ if and only if~$(G,k)$ is a \yes-instance. The construction ensures that~$G'$ has a perfect matching between~$A'$ and~$B'$. Now put~$k' := \Delta + k + 1 < 4 + 3n \leq \frac{|A'|}{2}$. It is easy to see that~$|A'|$ is even. The resulting instance~$(G', A', B', k')$ of \ElimCobipartite therefore satisfies all constraints. As~$G'$ has an elimination order of cost at most~$|A'| + k'$ if and only if~$(G,k)$ has cutwidth at most three, this completes the proof.
\qed
\end{proof}

\subsection{The Construction} \label{section:construction}
We start by defining an appropriate polynomial equivalence relationship~\eqvr. Let all malformed instances be equivalent under~\eqvr, and let two valid instances of \ElimCobipartite be equivalent if they agree on the sizes of the partite sets and on the value of~$k$. This is easily verified to be a polynomial equivalence relation.

Now we define an algorithm that combines a sequence of equivalent inputs into a small output instance. As a constant-size \no-instance is a valid output when the input consists of solely malformed instances, in the remainder we assume that the inputs are well-formed. By duplicating some inputs, we may assume that the number of input instances~$t$ is a square, i.e.,~$t = r^2$ for some integer~$r$. An input instance can therefore be indexed by two integers in the range~$[r]$. Accordingly, let the input consist of instances~$(G_{i,j}, A_{i,j}, B_{i,j}, k_{i,j})$ for~$i,j \in [r]$, that are equivalent under~\eqvr. Thus the number of vertices is the same over all partite sets; let this be~$n = |A_{i,j}| = |B_{i,j}|$ for all~$i,j \in [r]$. Similarly, let~$k$ be the common target value for all inputs. For each partite set~$A_{i,j}$ and~$B_{i,j}$ in the input, label the vertices arbitrarily as~$a_{i,j}^1, \ldots, a_{i,j}^n$ (respectively~$b_{i,j}^1, \ldots, b_{i,j}^n$). We construct a cobipartite graph~$G'$ that 
expresses the \OR of all the inputs, as follows.

\begin{enumerate}
	\item For~$i \in [r]$ make a vertex set~$A'_i$ containing~$n$ vertices~$\hat{a}_i^1, \ldots, \hat{a}_i^n$. 
	\item For~$i \in [r]$ make a vertex set~$B'_i$ containing~$n$ vertices~$\hat{b}_i^1, \ldots, \hat{b}_i^n$.
	\item Turn~$\bigcup _{i\in[r]} A'_i$ into a clique. Turn~$\bigcup _{i\in[r]} B'_i$ into a clique.
	\item For each pair~$i,j$ with~$i,j \in [r]$, we embed the adjacency of~$G_{i,j}$ into~$G'$ as follows: for~$p, q \in [n]$ make an edge~$\{\hat{a}_i^p, \hat{b}_j^q\}$ if~$\{a_{i,j}^p, b_{i,j}^q\} \in E(G_{i,j})$.
\end{enumerate}
It is easy to see that at this point in the construction, graph~$G'$ is cobipartite. For any~$i,j \in [r]$ the induced subgraph~$G'[A'_i \cup B'_j]$ is isomorphic to~$G_{i,j}$ by mapping~$\hat{a}_i^\ell$ to~$a_{i,j}^\ell$ and~$\hat{b}_j^\ell$ to~$b_{i,j}^\ell$. As~$G_{i,j}$ has a perfect matching between~$A_{i,j}$ and~$B_{i,j}$ by the definition of \ElimCobipartite, this implies that~$G'$ has a perfect matching between~$A'_i$ and~$B'_j$ for all~$i,j \in [r]$. These properties will be maintained during the remainder of the construction.
\begin{enumerate}[resume]
  \item For each~$i \in [r]$, add the following vertices to~$G'$:
  \begin{itemize}
  	\item $n$ \emph{checking vertices}~$C'_i = \{c^1_i, \ldots, c^n_i\}$, all adjacent to~$B'_i$.
	\item $n$ \emph{dummy vertices}~$D'_i = \{d^1_i, \ldots, d^n_i\}$, all adjacent to~$\bigcup _{j\in[r]} A'_j$ and to~$C'_i$.
	\item $\frac{n}{2}$ \emph{blanker vertices}~$X'_i = \{x^1_i, \ldots, x^{n/2}_i\}$, all adjacent to~$A'_i$.
  \end{itemize}
 	\item Turn~$\bigcup _{i \in [r]} A'_i \cup C'_i$ into a clique~$A'$. Turn~$\bigcup _{i \in [r]} B'_i \cup D'_i \cup X'_i$ into a clique~$B'$.
\end{enumerate}
The resulting graph~$G'$ is cobipartite with partite sets~$A'$ and~$B'$. Define~$k' := 3rn + \frac{n}{2} + k$. Observe that~$|A'| = 2rn$ and that~$|B'| = 2rn + \frac{rn}{2}$. Graph~$G'$ can easily be constructed in time polynomial in the total size of the input instances.

\textbf{Intuition.} Let us discuss the intuition behind the construction before proceeding to its formal analysis. To create a composition, we have to relate elimination orders in~$G'$ to those for input graphs~$G_{i,j}$. All adjacency information of the input graphs~$G_{i,j}$ is present in~$G'$. As~$A'$ is a clique in~$G'$, by \lemmaref{lemma:eliminateOnePartiteSet} there is a minimum-cost elimination order for~$G'$ that starts by eliminating all of~$B'$. But when eliminating vertices of some~$B'_{j^*}$ from~$G'$, they interact simultaneously with all sets~$A'_i$ ($i \in [r]$), so the cost of those eliminations is not directly related to the cost of elimination orders of a particular instance~$G_{i^*,j^*}$. We therefore want to ensure that low-cost elimination orders for~$G'$ first ``blank out'' the adjacency of~$B'$ to all but one set~$A'_{i^*}$, so that the cost of afterwards eliminating~$B'_{j^*}$ tells us something about the cost of eliminating~$G'_{i^*,j^*}$. To blank out the other adjacencies, we need 
earlier eliminations to make~$B'$ adjacent to all vertices of~$\bigcup _{i \in [r] \setminus \{i^*\}} A'_i$. These adjacencies will be created by eliminating the \emph{blanker} vertices. For an index~$i \in [r]$, vertices in~$X'_i$ are adjacent to~$A'_i$ and all of~$B'$. Hence eliminating a vertex in~$X'_i$ indeed blanks out the adjacency of~$B'$ to~$A'_i$. The weights of the various groups (simulated by duplicating vertices with identical closed neighborhoods) have been chosen such that low-cost elimination orders of~$G'$ starting with~$B'$, have to eliminate $r-1$ blocks of blankers~$X'_{i_1}, \ldots, X'_{i_{r-1}}$ before eliminating any other vertex of~$B'$. This creates the desired blanking-out effect. The checking vertices~$C'_i$ ($i \in [r]$) enforce that after eliminating $r-1$ blocks of blankers, an elimination order cannot benefit by mixing vertices from two or more sets~$B'_i, B'_{i'}$: each set~$B'_i$ from which a vertex is eliminated, introduces new adjacencies between~$B'$ and~$C'_i$. Finally, 
the dummy vertices are used to ensure that after one set~$B'_i \cup D'_i$ is completely eliminated, the cost of eliminating the remainder is small because~$|B'|$ has decreased sufficiently.

\subsection{Properties of the Constructed Instance} \label{section:analyzeComposedInstance}

The following type of elimination orders of~$G'$ will be crucial in the proof.

\begin{definition} \label{definition:canonical}
Let~$i^*, j^* \in [r]$. An elimination order~$\pi'$ of~$G'$ is \emph{$(i^*,j^*)$-canonical} if~$\pi'$ eliminates~$V(G)$ in the following order:
\begin{enumerate}
	\item first all blocks of blanker vertices~$X'_i$ for~$i \in [r] \setminus \{i^*\}$, one block at a time,
	\item then the vertices of~$B'_{j^*}$, followed by dummies~$D'_{j^*}$, followed by blankers~$X'_{i^*}$, 
	\item alternatingly a block~$B'_i$ followed by the corresponding dummies~$D'_i$, until all remaining vertices of~$\bigcup _{i\in[r]} B'_i \cup D'_i$ have been eliminated,
	\item and finishes with the vertices~$\bigcup _{i\in[r]} A'_i \cup C'_i$ in arbitrary order.
\end{enumerate}
\end{definition}

\noindent \lemmaref{lemma:instanceDominatesCanonicalCost} shows that the crucial part of a canonical elimination order is its behavior on~$B'_{j^*}$. 

\begin{lemma} \label{lemma:instanceDominatesCanonicalCost}
Let~$\pi'$ be an $(i^*,j^*)$-canonical elimination order for~$G'$.
\begin{enumerate}
	\item No vertex that is eliminated before the first vertex of~$B'_{j^*}$ costs more than~$3rn$.\label{claim:costBeforeB}
	\item When a vertex of~$D'_{j^*} \cup X'_{i^*}$ is eliminated, its cost does not exceed~$3rn + \frac{n}{2}$. \label{claim:costAfterB}
	\item No vertex that is eliminated after~$X'_{i^*}$ costs more than~$3rn$. \label{claim:costAfterX}
\end{enumerate}
\end{lemma}
\begin{proof}
\ref{claim:costBeforeB} By \defref{definition:canonical}, all vertices eliminated before~$B'_{j^*}$ are blanker vertices. The elimination of a vertex~$v$ in a block~$X'_i$ turns~$N(v)$ into a clique and removes~$v$. As~$A'$ and~$B'$ are cliques from the start, no extra edges can be introduced between members of~$A'$ or between members of~$B'$. When considering the effects of eliminating vertices from~$B'$, we therefore only have to consider which vertices of~$B'$ become adjacent to vertices in~$A'$. As blanker vertices are not adjacent to checking vertices, no elimination of a blanker vertex introduces adjacencies to sets~$C'_j$ for any~$j \in [r]$. Eliminating~$X'_i$ effectively makes all remaining vertices in~$B'$ adjacent to~$A'_i$, as~$X'_i$ and~$A'_i$ are adjacent by construction. With these insights we prove the first item of the lemma.

Consider the situation when~$0 \leq \ell < r-1$ blocks of blankers~$X'_{i_1}, \ldots, X'_{i_\ell}$ have already been eliminated, and we are about to eliminate a blanker vertex~$v_B$ in the next block~$X'_{i_{\ell+1}}$. Then~$N[v]$ contains all remaining vertices in~$B'$, of which there are~$|B'| - \ell\cdot \frac{n}{2} = 2rn + \frac{(r-\ell)n}{2}$. Now consider the neighborhood of~$v_B$ in~$A'$. As observed above,~$N[v_B]$ does not contain checking vertices. When it comes to adjacencies into~$\bigcup _{i \in [r]} A'_i$, vertex~$v_B$ in block~$X'_{i_{\ell+1}}$ is adjacent to~$A'_{i_{\ell+1}}$, and to the blocks~$A'_{i_1}, \ldots, A'_{i_\ell}$ for the blankers~$X'_{i_1}, \ldots, X'_{i_{\ell}}$ that have previously been eliminated. Hence~$v_B$ has~$(\ell+1)n$ neighbors in~$A'$. Summing up the contributions from the two partite sets, we find~$|N[v]| = 2rn + \frac{(r - \ell)n}{2} + (\ell+1)n = 3rn + \frac{(\ell - r)n}{2} + n$. The largest value is attained when the last block of blankers unequal to~$X'_{i^*}$ is 
about to be eliminated; at that point~$\ell = r-2$ blocks have been eliminated already, which results in a cost of~$3rn$ for the first vertex of the last block that is eliminated. The other vertices~$X'_{i_{\ell + 1}}$ are eliminated immediately after~$v$, by \defref{definition:canonical}. Hence their closed neighborhood at time of elimination is smaller than that of~$v$: elimination of~$v$ does not introduce any new adjacencies for vertices with the same closed neighborhood as~$v$. Hence the remaining vertices of~$X'_{i_{\ell+1}}$ cost less than~$v$. Thus the cost of eliminating the vertices before~$B'_{j^*}$ does not exceed~$3rn$.

\ref{claim:costAfterB} Let~$G'_B$ be the graph that is obtained from~$G'$ by eliminating according to~$\pi'$ until just after the last vertex of~$B'_{j^*}$. Then~$G'_B$ contains exactly one block of blankers~$X'_{i^*}$, as all other sets have been eliminated before~$B'_{j^*}$. It does not contain~$B'_{j^*}$ as it was just eliminated. The elimination of~$B'_{j^*}$ has made the remainder of~$B'$ adjacent to~$\bigcup _{i\in[r]} A'_i$, as~$B'_{j^*}$ has a perfect matching into~$A'_i$ for all~$i \in [r]$. According to \defref{definition:canonical}, elimination order~$\pi'$ eliminates~$D'_{j^*}$ just after~$B'_{j^*}$. At that point, the neighborhood of the dummy vertices~$D'_{j^*}$ into~$\bigcup _{j \in[r]} C'_j$ is exactly~$C'_{j^*}$: they were initially adjacent, the eliminated blanker vertices were not adjacent to any checking vertices, and the eliminated vertices from~$B'_{j^*}$ see only the checking vertices~$C'_{j^*}$, by construction. Hence the cost of eliminating the first dummy vertex in~$D'_{j^*}$ is~$|\
bigcup_{j \in [r] \setminus \{j^*\}} B'_{j^*} \cup D'_{j^*}| + |D'_{j^*}| + |X'_{i^*}| + |\bigcup _{i\in[r]} A'_i| + |C'_{j^*}|$, which is~$2(r-1)n + n + \frac{n}{2} + rn + n = 3rn + \frac{n}{2}$. As the other dummy vertices in~$D'_{j^*}$ have exactly the same closed neighborhood, their elimination is not more expensive.

By \defref{definition:canonical}, order~$\pi'$ follows the elimination of~$D'_{j^*}$ by eliminating~$X'_{i^*}$. At the time of elimination, $N[X'_{i^*}] \cap B'$ has size~$|\bigcup _{j \in[r] \setminus \{j^*\}} B'_j \cup D'_j| + |X'_{i^*}| = 2(r-1)n + \frac{n}{2}$. Vertices in~$X'_{i^*}$ are adjacent to~$\bigcup _{i\in[r]} A'_i$, and to exactly one set of checking vertices, namely~$C'_{j^*}$; the elimination of~$B'_{j^*}$ has introduced these adjacencies. Hence the cost of eliminating the first vertex of~$X'_{i^*}$ is~$2(r-1)n + \frac{n}{2} + rn + n = 3rn - \frac{n}{2}$. As the other vertices in~$X'_{i^*}$ have the same neighborhood, they are not more expensive. In summary, no vertex of~$D'_{j^*} \cup X'_{i^*}$ costs more than~$3rn + \frac{n}{2}$ when eliminated.

\ref{claim:costAfterX} Let~$G'_X$ be the graph that is obtained from~$G'$ by eliminating according to~$\pi'$ until just after the last vertex of~$X'_{i^*}$. Then~$G'_X$ does not contain any blanker vertices, as all such sets have been eliminated. Similarly, it does not contain~$B'_{j^*}$ or~$D'_{j^*}$. The eliminations up until~$X'_{i^*}$ have made all of~$B'$ adjacent to~$\bigcup _{i\in[r]} A'_i$. Vertices in a set~$B'_j \cup D'_j$ for~$j \neq j^*$ are adjacent to the checking vertices~$C'_{j}$ (by construction) and~$C'_{j^*}$ (because the elimination of~$B'_{j^*}$ introduced these adjacencies), but to no other checking vertices. Now consider the first vertex~$v$ that is eliminated after~$X'_{i^*}$; by \defref{definition:canonical} it is contained in some set~$B'_j$ with~$j \neq j^*$. As no blanker vertex remains, and~$B'_{j^*} \cup D'_{j^*}$ have been eliminated, there are exactly~$2rn - 2n$ vertices left in~$B'$. Vertex~$v$ is adjacent to all~$rn$ vertices in~$\bigcup _{i\in[r]} A'_i$, to~$C'_{j^*}$ and 
to the checking vertices corresponding to its own index. Hence the cost of~$v$ is~$2rn-2n + rn + n + n= 3rn$. The cost of the succeeding blankers~$D'_j$ in the same block is not more than that of~$v$. 

When eliminating the next group~$B'_{j'}$, observe that~$2n$ neighbors have been lost in~$B'$ (the set~$B'_j \cup D'_j$ that was eliminated), whereas only~$n$ new neighbors have been introduced (the set~$C'_j$). Hence the cost of later groups of vertices~$B'_{j'}$ does not exceed the cost of~$v$, and so does not exceed~$3rn$. Finally, when all of~$B'$ has been eliminated then only the vertex set~$A'$ of size~$2rn$ remains. At that point, no vertex can have cost more than~$2rn$ as there are only~$2rn$ vertices left in the graph. Thus the cost of eliminating~$A'$ satisfies the claimed bound, after which the entire graph is eliminated.
\qed
\end{proof}

The next lemma links this behavior to the cost of a related elimination order for~$G_{i^*,j^*}$. Some terminology is needed. Consider an $(i^*,j^*)$-canonical elimination order~$\pi'$ for~$G'$, and an elimination order~$\pi$ for~$G_{i^*,j^*}$ that eliminates all vertices of~$B_{i^*,j^*}$ before any vertex of~$A_{i^*,j^*}$. By numbering the vertices in~$B_{i^*,j^*}$ (a partite set of~$G_{i^*,j^*}$) from~$1$ to~$n$, we created a one-to-one correspondence between~$B_{i^*,j^*}$ and~$B'_{j^*}$, the first set of non-blanker vertices eliminated by~$\pi'$. Hence we can compare the relative order in which vertices of~$B_{i^*,j^*}$ are eliminated in~$\pi$ and~$\pi'$. If both~$\pi$ and~$\pi'$ eliminate the vertices of~$B_{i^*,j^*}$ in the same relative order, then we say that the elimination orders \emph{agree on~$B_{i^*,j^*}$}.

\begin{lemma} \label{lemma:canonicalOrderCorresponds}
Let~$\pi'$ be an $(i^*,j^*)$-canonical elimination order of~$G'$. Let~$\pi$ be an elimination order for~$G_{i^*,j^*}$ that eliminates all vertices of $B_{i^*,j^*}$ before any vertex of~$A_{i^*,j^*}$. If~$\pi'$ and~$\pi$ agree on~$B_{i^*,j^*}$, then~$c_{G'}(\pi') = 3rn + \frac{n}{2} - n + c_{G_{i^*,j^*}}(\pi)$.
\end{lemma}
\begin{proof}
Consider the graph~$G'_B$ obtained from~$G'$ by performing the eliminations according to~$\pi'$ until we are about to eliminate the first vertex of~$B'_{j^*}$. By \defref{definition:canonical} this means that all blocks of blankers~$X'_j$ for~$j \neq j^*$ have been eliminated, and no other vertices. Using the construction of~$G'$ it is easy to verify that these eliminations have made all remaining vertices of~$B'$ adjacent to~$\bigcup _{i \in [r] \setminus \{i^*\}} A'_i$, and that no new adjacencies have been introduced to~$\bigcup _{i \in [r]} C'_i$ or to~$A'_{i^*}$. Graph~$G'[A'_{i^*} \cup B'_{j^*}]$ was initially isomorphic to~$G_{i^*,j^*}$ by the obvious isomorphism based on the numbers assigned to the vertices. As no vertex adjacent to~$A'_{i^*}$ has been eliminated yet, this also holds for~$G'_B[A'_{i^*} \cup B'_{j^*}]$.

Consider what happens when eliminating the first vertex~$v'$ of~$B'_{j^*}$ according to~$\pi'$. Let~$v \in B_{i^*,j^*}$ be the corresponding vertex in~$G_{i^*,j^*}$. By the fact that the elimination orders agree,~$v$ is the first vertex of~$B_{i^*,j^*}$ to be eliminated under~$\pi$.

The set~$N_{G'_B}[v']$ contains~$C'_{j^*}$, $\bigcup_{j \neq j^*} B'_j \cup D'_j$, $\bigcup _{i \neq i^*} A'_i$, $X'_{i^*}$, $D'_{j^*}$, and the vertices of~$G'[A'_{i^*} \cup B'_{j^*}]$ that correspond exactly to~$N_{G_{i^*,j^*}}[v]$ by the isomorphism. So the cost of eliminating~$v'$ from~$G'$ exceeds the cost of eliminating~$v$ from~$G_{i^*,j^*}$ by exactly~$|C'_{j^*}| + |\bigcup_{j \neq j^*} B'_j \cup D'_j| + |\bigcup _{i \neq i^*} A'_i| + |X'_{i^*}| + |D'_{j^*}| = n + 2(r-1)n + (r-1)n + \frac{n}{2} + n = 3rn + \frac{n}{2} - n$.
Now observe that by the isomorphism, eliminating~$v'$ from~$G'$ has exactly the same effect on the neighborhoods of~$B'_{j^*}$ into~$A'_{i^*}$, as eliminating~$v$ from~$G_{i^*,j^*}$ has on the neighborhoods of~$B_{i^*,j^*}$ into~$A_{i^*,j^*}$. Thus after one elimination, the remaining vertices of~$A'_{i^*} \cup B'_{j^*}$ and~$A_{i^*,j^*} \cup B_{i^*,j^*}$ induce subgraphs of~$G'$ and~$G_{i^*,j^*}$ that are isomorphic. Hence we may apply the same argument to the next vertex that is eliminated. Repeating this argument we establish that for each vertex in~$B'_{j^*}$, its elimination from~$G'$ costs exactly~$3rn + \frac{n}{2} - n$ more than the corresponding elimination in~$G_{i^*,j^*}$.

Now consider the cost of~$\pi$ on~$G_{i^*,j^*}$: it is at least~$n+1$, as the first vertex to be eliminated is adjacent to all of~$B_{i^*,j^*}$ (the graph is cobipartite) and to at least one vertex of~$A_{i^*,j^*}$ (since the \ElimCobipartite instance~$G_{i^*,j^*}$ has a perfect matching between its two partite sets). After all vertices of~$B_{i^*,j^*}$ have been eliminated from~$G_{i^*,j^*}$, the remaining vertices cost at most~$n$; there are at most~$n$ vertices left in the graph at that point. Hence the cost of~$\pi$ on~$G_{i^*,j^*}$ is determined by the cost of eliminating~$B_{i^*,j^*}$. For each vertex from that set that is eliminated,~$\pi'$ incurs a cost exactly~$3rn + \frac{n}{2} - n$ higher. Hence~$c_{G'}(\pi')$ is at least~$(3rn + \frac{n}{2} - n) + (n+1) = 3rn + \frac{n}{2} + 1$. By \lemmaref{lemma:instanceDominatesCanonicalCost} the cost that~$\pi'$ incurs before eliminating the first vertex of~$B'_{j^*}$ is at most~$3rn$, the cost of eliminating~$D'_{j^*} \cup X'_{i^*}$ is at most~$3rn + \frac{n}
{2}$, and the cost incurred after eliminating the last vertex of~$B'_{j^*}$ is at most~$3rn$. Hence the cost of~$\pi'$ is determined by the cost of eliminating the vertices of~$B'_{j^*}$. As this is exactly~$3rn + \frac{n}{2} - n$ more than the cost of~$\pi$ on~$G_{i^*,j^*}$, this proves the lemma.
\qed
\end{proof}

The last technical step of the proof is to show that if~$G'$ has an elimination order of cost at most~$k'$, then it has such an order that is canonical.

\begin{lemma} \label{lemma:exists:canonical:opt}
If $G'$ has an elimination order of cost at most~$k'$, then there are indices~$i^*,j^* \in [r]$ such that~$G'$ has an $(i^*,j^*)$-canonical elimination order of cost at most~$k'$.
\end{lemma}
\begin{proof}
Let~$\pi'$ be an elimination order for~$G'$ of cost at most~$k'$. As~$A'$ is a clique in~$G'$, we may assume by \lemmaref{lemma:eliminateOnePartiteSet} that~$\pi'$ eliminates all vertices of~$B'$ before any vertex of~$A'$ (Property 1). As each set~$X'_i$ forms a block in~$G'$, by \obsref{observation:finishBlock} we can adapt~$\pi'$ such that it eliminates the vertices of a set~$X'_i$ contiguously for all~$i \in [r]$ (Property 2). Note that for~$i \in [r]$ and vertices~$d \in D'_i$ and~$u \in B'_i$, we have~$N_G[u] \subseteq N_G[d]$ by construction. Hence by \obsref{observation:finishBlock} we may assume that when a dummy~$d \in D'_i$ is about to be eliminated, all vertices of the corresponding set~$B'_i$ are already eliminated (Property 3). Using these structural properties we proceed with the proof.

Consider the process of eliminating~$G'$ by~$\pi'$. At some point,~$\pi'$ has eliminated $r-2$ distinct blocks of blankers~$X'_{i_1}, \ldots, X'_{i_{r-2}}$. Consider the first vertex~$v_B$ of the blanker~$X'_{i_{r-1}}$ that is eliminated after that point, and note that possibly non-blanker vertices are eliminated in between. Let~$G'_B$ be the graph obtained from~$G'$ by eliminating all vertices before~$v_B$. 

\begin{numberedclaim} \label{claim:firstEliminateMostBlankers}
Graph~$G'_B$ still contains the vertices~$\bigcup _{i\in[r]} B'_i \cup D'_i$: no vertex in this set is eliminated before~$v_B$.
\end{numberedclaim}
\begin{proof}
Assume for a contradiction that some vertex of~$\bigcup _{i\in[r]} B'_i \cup D'_i$ is eliminated before~$v_B$. We first show how to derive a contradiction when there is an index~$i \in [r]$ such that~$B'_i$ is eliminated completely before~$v_B$ (Case 1). Afterwards we show how to derive a contradiction when at least one vertex of~$B'_i$ remains in~$G'_B$ for all~$i \in [r]$ (Case 2).

\textbf{Case 1.} If there is an index~$j$ such that no vertex of~$B'_j$ remains in~$G'_B$, then let~$j^*$ be the index of the first set~$B'_j$ to be eliminated from~$G'$ completely. Consider the moment when the last vertex~$u$ of~$B_{j^*}$ is eliminated. By our choice of~$j^*$ and Property 3, we know that no dummy vertex has been eliminated yet. So consider the closed neighborhood of~$u$ at the moment of its elimination. It contains all~$rn$ dummy vertices. As at least two blocks of blankers remain,~$N[u]$ contains at least~$\frac{2n}{2}$ blanker vertices. We claim that~$u$ has become adjacent to all vertices of~$\bigcup _{i\in[r]} A'_i$. To see this, recall that~$u$ was the last vertex of~$B_{j^*}$ to be eliminated. As we observed during the construction of~$G'$, there is a perfect matching between~$A'_i$ and~$B'_{j^*}$ for all~$i \in [r]$. Hence for each vertex in~$\bigcup _{i\in[r]} A'_i$, if~$u$ was not originally adjacent to it, then it has become adjacent to it by eliminating the vertex of~$B'_{j^*}$ 
that was matched to it. Thus~$u$ is indeed adjacent to all of~$\bigcup _{i\in[r]} A'_i$. For each vertex in a set~$B'_j$ with~$j \neq j^*$ that is eliminated before~$u$, the elimination has made~$u$ adjacent to~$C'_j$. By our choice of~$j^*$, no such set~$B'_j$ is eliminated completely. Hence for each set~$B'_j$ with~$j \neq j^*$ from which (less than~$n$) vertices were eliminated,~$u$ has picked up~$n$ new neighbors in the set~$C'_j$. So the number of neighbors of~$u$ in~$\bigcup_{j \in [r] \setminus \{j^*\}} B'_j \cup C'_j$ is at least~$(r-1)n$. Adding up the contribution of the blankers, the dummies, of~$\bigcup _{i\in[r]} A'_i$, of~$\bigcup_{j \in [r] \setminus \{j^*\}} B'_j \cup C'_j$, and of~$C'_{j^*}$, to~$N[u]$, we find that~$|N[u]| \geq \frac{2n}{2} + rn + rn + (r-1)n + n \geq 3rn + n$. This value exceeds~$k'$, as~$k < \frac{n}{2}$ by the definition of \ElimCobipartite. Hence we find a contradiction to the assumption that~$\pi'$ has cost at most~$k'$.

\textbf{Case 2.} Assume now that for each~$j \in [r]$ at least one vertex of~$B'_j$ remains in~$G'_B$, which implies by Property 3 that all dummies are present in~$G'_B$. Recall that we assumed, for a contradiction, that some vertex~$u$ of~$\bigcup _{j\in[r]} B'_j \cup D'_j$ was eliminated before~$v_B$. As~$u$ is no dummy, it is contained in some set~$B'_{j^*}$. By the adjacency of~$B'_{j^*}$ to~$C'_{j^*}$, the elimination has made the blanker~$X'_{i_{r-1}}$ adjacent to~$C'_{j^*}$. We will show that this causes the cost of~$v_B$ to exceed~$k'$. 
	
To see this, consider the neighbors of~$v_B$ in the various sets. For each set~$B'_j$ from which  vertices were eliminated, we have eliminated less than~$n$ vertices (at least one vertex remains by the precondition to this case). For those sets~$B'_j$, the blankers~$X'_{i_{r-1}}$ have picked up adjacencies to the corresponding checkers~$C'_j$. Thus~$|N[v_B] \cap (\bigcup_{j \in[r] \setminus \{j^*\}} B'_j \cup C'_j)| \geq (r-1)n$. As~$v_B$ is the first blanker vertex to be eliminated after~$r-2$ blocks of blankers were already eliminated, there are two blocks of blankers left, giving~$\frac{2n}{2}$ vertices in~$N[v_B] \cap (\bigcup _{i \in [r]} X'_i)$. The prior eliminations of blankers~$X'_{i_1}, \ldots, X'_{i_{r-2}}$ made~$X'_{i_{r-1}}$ adjacent to the corresponding sets~$A'_{i_1}, \ldots, A'_{i_{r-2}}$, and by construction~$v_B \in X'_{i_{r-1}}$ is adjacent to~$A'_{i_{r-1}}$. Now consider the remaining index~$i_r \in [r] \setminus \{i_1, \ldots, i_{r-1}\}$, and let~$i^* := i_r$ for convenience.

Recall that~$B'_{j^*}$ has a perfect matching into~$A'_{i^*}$ by the construction of~$G'$. Hence for each vertex vertex~$u$ that was eliminated from~$B'_{j^*}$, vertex~$v_B$ has become adjacent to~$u$'s matching partner in the set~$A'_{i^*}$. Hence, letting~$\ell$ denote the number of vertices eliminated from~$B'_{j^*}$, we know that~$v_B$ is adjacent to at least~$\ell$ vertices in~$A'_{i^*}$. Summing up the contributions of the blankers, the dummies, the set~$(\bigcup_{i \in[r] \setminus \{j^*\}} B'_i \cup C'_i)$, the set~$C'_{j^*}$, the set~$\bigcup _{i \in [r] \setminus \{i^*\}} A'_i$, and the set~$A'_{i^*} \cup B'_{j^*}$ to~$|N[v_B]|$, we find that~$|N[v_B]| \geq \frac{2n}{2} + rn + (r-1)n + n + (r-1)n + n \geq 3rn + n > k'$, which is a contradiction to the assumption that the cost of~$\pi'$ is at most~$k'$.

As the two cases are exhaustive, we have established that when~$v_B$ is eliminated, all vertices of~$\bigcup _{i\in[r]} B'_i \cup D'_i$ still remain in the graph~$G'_B$. 
\claimqed
\end{proof}

We need two more claims to complete the proof of \lemmaref{lemma:exists:canonical:opt}. As~$\pi'$ is block-contiguous with respect to the blankers (Property 2), after~$v_B$ it eliminates the rest of~$X'_{i_{r-1}}$. Afterwards only a single group of blankers remains, say~$X'_{i_r}$.

\begin{numberedclaim} \label{claim:thenEliminateInstance}
After eliminating~$X'_{i_{r-1}}$, order~$\pi'$ eliminates a vertex in~$\bigcup _{i\in[r]} B'_i$.
\end{numberedclaim}
\begin{proof}
By Property 1, all vertices of~$B'$ are eliminated before any vertex of~$A'$. Recall that~$B'$ consists of blankers~$\bigcup_{i\in[r]} X'_i$ and the vertices~$\bigcup _{i\in[r]} B'_i \cup D'_i$. As~$X'_{i_{r-1}}$ is the $r-1$-th block of blankers to be eliminated, afterwards the only vertices in~$B'$ remaining are~$X'_{i_r}$ and~$\bigcup _{i \in [r]} B'_i \cup D'_i $. By Property 3,~$\pi'$ eliminates all vertices of~$B'_i$ before eliminating a dummy in the corresponding set~$D'_i$. Hence if~$\pi'$ does not follow the elimination of~$X'_{i_{r-1}}$ by a vertex of~$\bigcup _{i\in[r]} B'_i$, it eliminates~$X'_{i_r}$. If this is the case, then all blankers have been eliminated before eliminating any vertex of~$\bigcup _{i \in [r]} B'_i \cup D'_i$. Now consider the first vertex~$u$ of~$\bigcup _{i\in[r]} B'_i$ that is eliminated, and suppose it is contained in~$B'_{j^*}$. Eliminating all blankers has made~$B'_{j^*}$ adjacent to all of~$\bigcup_{i\in[r]} A'_i$. By construction~$B'_{j^*}$ is adjacent to the~$n$ 
checking vertices~$C'_{j^*}$. By Property 3 it is adjacent to all dummies. Summing up the contributions of the dummies, of~$\bigcup _{i\in[r]} B'_i$, of~$\bigcup _{i \in [r]} A'_i$, and of the single set~$C'_{j^*}$, to~$N[u]$, we find that the cost of~$u$ is at least~$rn + rn + rn + n > k'$; a contradiction.
\claimqed
\end{proof}

Before proving the next claim, we make an observation. Let~$u$ be the first vertex of~$\bigcup _{i \in [r]} B'_i$ that is eliminated by~$\pi'$, and suppose that~$u \in B'_{j^*}$. The elimination of~$u$ makes the last group of blankers~$X'_{i_r}$ adjacent to the checking vertices~$C'_{j^*}$, as~$B'_{j^*}$ is adjacent to~$C'_{j^*}$. This implies that after the elimination of~$u \in B'_{j^*}$, the closed neighborhood of~$X'_{i^*}$ is a superset of the closed neighborhood of a remaining vertex in~$B'_{j^*}$. To see this, note that at that stage,~$X'_{i^*}$ is adjacent to the remainder of~$B'$, to~$\bigcup_{i \in [r] \setminus \{i^*\}} A'_i$ (by eliminating the previous blankers), to~$A'_{i^*}$ (by construction), and to~$C'_{j^*}$ (by eliminating~$u$). On the other hand, vertices in~$B'_{j^*}$ see the remainder of~$B'$, they see~$\bigcup_{i \in [r] \setminus \{i^*\}} A'_i$, a subset of~$A'_{i^*}$ that depends on the edges in the graph~$G_{i^*,j^*}$, and~$C'_{j^*}$. 
Hence, by the same reasoning as in \obsref{observation:finishBlock}, if a vertex~$z \in X'_{i^*}$ is eliminated after the \emph{first} vertex of~$B'_{j^*}$ (i.e.,~$u$) but before the \emph{last} vertex of~$B'_{j^*}$, then the cost of~$\pi'$ does not increase when eliminating all vertices of~$B'_{j^*}$ just before~$z$. Hence we may assume that~$\pi'$ eliminates all of~$B'_{j^*}$ before any vertex of~$X'_{i^*}$; we call this Property 4. We use this in the proof of the following claim.

\begin{numberedclaim} \label{claim:finishOneInstance}
All vertices of~$B'_{j^*}$ are eliminated before any vertex of~$\bigcup _{j \in [r] \setminus \{j^*\}} B'_j$.
\end{numberedclaim}
\begin{proof}
By Property 4, all vertices of~$B'_{j^*}$ are eliminated before the last blanker~$X'_{i^*}$. Now suppose that before eliminating the last vertex of~$B'_{j^*}$, order~$\pi'$ eliminates some vertex~$v \in B'_{j'}$ with~$j' \neq j^*$. Let~$v$ be the first vertex with this property. By Property 4, all vertices in~$X'_{i^*}$ remain in the graph when~$v$ is eliminated. This causes the cost of~$v$ to exceed~$k'$. To see this, observe that at the time of elimination, the closed neighborhood of~$v$ contains all~$rn$ dummy vertices (by Property 3), it contains the~$\frac{n}{2}$ vertices of~$X'_{i^*}$, it contains~$C'_{j^*}$ (by elimination of~$u$) and~$C'_{j'}$ (by construction), which contain~$n$ vertices each. Additionally,~$N[v]$ contains~$\bigcup _{i\in[r] \setminus \{j^*\}} B'_i$ by our choice of~$v$, and~$\bigcup _{i \in [r] \setminus \{i^*\}} A'_i$ by the eliminations of earlier groups of blankers, for a subtotal of~$rn + \frac{n}{2} + 2n + (r-1)n + (r-1)n = 3rn + \frac{n}{2}$. If~$\ell$ vertices have been 
eliminated from~$B'_{j^*}$ prior to elimination of~$v$, then~$N[v]$ contains~$n - \ell$ vertices from~$B'_{j^*}$, but has gained~$\ell$ neighbors in~$A'_{i^*}$ by the perfect matching between~$A'_{i^*}$ and~$B_{j^*}$ in~$G'$. Hence the remaining vertices in~$A'_{i^*} \cup B'_{j^*}$ contribute at least~$\ell + (n - \ell)$ vertices to the cost of~$v$. Thus the cost of~$v$ is at least~$3rn + \frac{n}{2} + n$, which is more than~$k'$; a contradiction.
\claimqed
\end{proof}

Using Claims~\ref{claim:firstEliminateMostBlankers},~\ref{claim:thenEliminateInstance}, and~\ref{claim:finishOneInstance}, we prove \lemmaref{lemma:exists:canonical:opt}. By Property 1, elimination order~$\pi'$ eliminates~$B'$ before~$A'$. By \claimref{claim:firstEliminateMostBlankers},~$\pi'$ did not eliminate any vertex of~$\bigcup _{i \in [r]} B'_i \cup D'_i$ when the first vertex of the $r-1$-th block of blankers is eliminated. As~$\pi'$ is block-contiguous with respect to the blankers, its initial behavior matches that of a canonical elimination order (\defref{definition:canonical}): it eliminates~$r-1$ distinct blocks of blankers~$X'_{i_1}, \ldots, X'_{i_{r-1}}$ before any vertex of~$\bigcup _{j\in[r]} B'_j \cup D'_j$. By \claimref{claim:thenEliminateInstance} it then eliminates a vertex of~$\bigcup _{j\in[r]} B'_j$, say a vertex in~$B'_{j^*}$. By \claimref{claim:finishOneInstance} it completes the elimination of~$B'_{j^*}$ before touching vertices in~$\bigcup _{j \in [r] \setminus \{j^*\}} B'_j$, by 
Property 3 it eliminates~$B'_{j^*}$ before any dummy, and by Property 4 it eliminates~$B'_{j^*}$ before the last blanker~$X'_{i_r}$. Hence after the~$r-1$ blocks of blankers, the vertices of~$B'_{j^*}$ are eliminated consecutively. 

Once this is done, the closed neighborhoods of~$D'_{j^*}$ and~$X'_{i^*}$ coincide: by the perfect matchings between~$B'_{j^*}$ and~$A'_i$ (for all~$i \in [r]$) in~$G'$, eliminating all of~$B'_{j^*}$ made~$D'_{j^*}$ and~$X'_{i^*}$ adjacent to~$\bigcup _{i\in[r]}A'_i$. Furthermore,~$N[D'_{j^*}] \cap \bigcup _{i \in [r]} C'_i = N[X'_{i^*}] \cap \bigcup _{i \in [r]} C'_i = C'_{j^*}$: the dummies see~$C'_{j^*}$ by construction, while~$X'_{i^*}$ sees it because of the elimination of~$B'_{j^*}$. The closed neighborhoods of~$D'_{j^*}$ and~$X'_{i^*}$ are subsets of the closed neighborhoods of the other vertices that remain in~$B'$ at that point: vertices in a set~$B'_j \cup D'_j$ for~$j \neq j^*$ see~$\bigcup _{i \in[r]}A'_i$ together with both~$A'_{j^*}$ and~$A'_j$, while the latter set is not seen by~$D'_{j^*} \cup X'_{i^*}$. Hence by \obsref{observation:finishBlock} we may assume that after finishing~$B'_{j^*}$, order~$\pi'$ eliminates~$D'_{j^*}$ followed by~$X'_{i^*}$. 

Once that is done, the only vertices remaining in~$B'$ are $\bigcup _{i \in [r] \setminus \{j^*\}} B'_i \cup D'_i$. It is easy to see that for any~$j \in [r] \setminus \{j^*\}$, all vertices in~$B'_j \cup D'_j$ have the same closed neighborhood at that stage, consisting of the remainder of~$B'$ together with~$C'_j \cup C'_{i^*}$ and~$\bigcup _{i \in [r]} A'_i$. By \obsref{observation:finishBlock} we may assume that~$\pi'$ is block-contiguous after eliminating~$X'_{i^*}$, which means it eliminates the sets~$B'_j \cup D'_j$ one at a time. As we may shuffle the order within a set~$B'_j \cup D'_j$ without changing the cost (all closed neighborhoods of vertices from such a set are identical), we may assume that the remaining actions of~$\pi'$ on~$B'$ are alternatingly eliminating a set~$B'_j$ followed by the corresponding set~$D'_j$, until all of~$B'$ is eliminated. Then~$\pi'$ finishes by eliminating~$A'$ in some order. As this form exactly matches the definition of an $(i^*,j^*)$-canonical elimination order, we 
have proved that whenever an elimination order of~$G'$ exists that has cost at most~$k'$, then there is one that is canonical. This proves \lemmaref{lemma:exists:canonical:opt}.
\qed
\end{proof}

\subsection{Proof of Theorem \ref{theorem:treewidth:sparsification}} \label{section:sparsificationProof}

Having analyzed the relationship between elimination orders for~$G'$ and for the input graphs~$G_{i,j}$ ($i, j \in [r]$), we can complete the proof. By combining the previous lemmata it is easy to show that~$G'$ acts as the logical \OR of the inputs.

\begin{lemma} \label{lemma:constructionGivesOr}
$G'$ has an elimination order of cost~$\leq k'$ $\Leftrightarrow$ there are~$i,j \in [r]$ such that~$G_{i,j}$ has an elimination order of cost~$\leq n + k$.
\end{lemma}
\begin{proof}
($\Rightarrow$) Assume that~$G'$ has an elimination order~$\pi'$ of cost at most~$k'$. By \lemmaref{lemma:exists:canonical:opt} we may assume that~$\pi'$ is $(i^*,j^*)$-canonical, for appropriate choices of~$i^*$ and~$j^*$. Build an elimination order~$\pi$ for~$G_{i^*,j^*}$ that agrees with~$\pi'$ on~$B_{i^*,j^*}$. By \lemmaref{lemma:canonicalOrderCorresponds} this shows that~$c_{G'}(\pi') = 3rn + \frac{n}{2} - n + c_{G_{i^*,j^*}}(\pi)$. Hence~$c_{G_{i^*,j^*}}(\pi) = c_{G'}(\pi') - 3rn - \frac{n}{2} + n \leq k' - 3rn - \frac{n}{2} + n = n + k$. Thus~$G_{i^*,j^*}$ has an elimination order of cost at most~$n + k$.

($\Leftarrow$) In the other direction, suppose that~$G_{i^*,j^*}$ has an elimination order~$\pi$ of cost at most~$n + k$. As~$A_{i^*,j^*}$ is a clique in~$G_{i^*,j^*}$, by \lemmaref{lemma:eliminateOnePartiteSet} we may assume that~$\pi$ eliminates all vertices of~$B_{i^*,j^*}$ before any vertex of~$A_{i^*,j^*}$. Using \defref{definition:canonical} it is easy to see that a canonical elimination order~$\pi'$ for~$G'$ exists that agrees with~$\pi$ on~$B_{i^*,j^*}$. By \lemmaref{lemma:canonicalOrderCorresponds} the cost of~$\pi'$ on~$G'$ exceeds the cost of~$\pi$ on~$G_{i^*,j^*}$ by exactly~$3n+\frac{n}{2} - n$. So the cost of~$\pi'$ on~$G'$ is at most~$3n + \frac{n}{2} - n + (n + k) = k'$, which proves this direction of the claim.
\qed
\end{proof}

\begin{lemma} \label{lemma:elimCrossComposesIntoTw}
There is an \OR-cross-composition of \ElimCobipartite into \nTreewidth of cost~$\sqrt{t}$.
\end{lemma}
\begin{proof}
In \sectref{section:construction} we gave a polynomial-time algorithm that, given instances $(G_{i,j}, A_{i,j}, B_{i,j}, k_{i,j})$ of \ElimCobipartite that are equivalent under~\eqvr for~$i,j \in [r]$, constructs a cobipartite graph~$G'$ with partite sets~$A'$ and~$B'$, and an integer~$k'$. By \lemmaref{lemma:constructionGivesOr} the resulting graph~$G'$ has an elimination order of cost~$k'$ if and only if there is a \yes-instance among the inputs. By the correspondence between treewidth and bounded-cost elimination orders of \thmref{theorem:tw:eliminationcharacterization}, this shows that~$G'$ has treewidth at most~$k' - 1$ if and only if there is a \yes-instance among the inputs. The polynomial equivalence relationship ensured that all partite sets of all inputs have the same number of vertices. For partite sets of size~$n$, the constructed graph~$G'$ satisfies~$|A'| = 2rn$ and~$|B'| = \frac{5rn}{2}$. The number of vertices in~$G'$ is~$n' = \frac{9rn}{2}$. Consider the \nTreewidth instance~$(G', n', k'-1)$.
 It expresses the logical \OR of a series of~$r^2 = t$ \ElimCobipartite instances using a parameter value of~$\frac{9n\sqrt{t}}{2} \in \Oh(n\sqrt{t})$. Hence the algorithm gives an \OR-cross-composition of \ElimCobipartite into \nTreewidth of cost~$\sqrt{t}$.
\qed
\end{proof}

\noindent \thmref{theorem:treewidth:sparsification} follows from the combination of \lemmaref{lemma:elimCrossComposesIntoTw}, \lemmaref{lemma:elimCobipartiteNPC}, and \thmref{theorem:orcrosscomposition:polynomialbound}. Since the pathwidth of a cobipartite graph equals its treewidth~\cite{Mohring96} and the graph formed by the cross-composition is cobipartite, the same construction gives an \OR-cross-composition of bounded cost into \nPathwidth.

\begin{corollary}
If \nPathwidth admits a (generalized) kernel of size $\Oh(n^{2 - \epsilon})$, for some~$\epsilon > 0$, then \containment.
\end{corollary}

\section{Quadratic-Vertex Kernel for Treewidth [VC]} \label{section:tw:kernel}
In this section we present an improved kernel for \TreewidthVC, which is formally defined as follows.

\parproblemdef{\TreewidthVC}
{A graph~$G$, a vertex cover~$X \subseteq V(G)$, and an integer~$k$.}
{$|X|$.}
{Is the treewidth of~$G$ at most~$k$?}

\noindent Our kernelization revolves around the following notion.

\begin{definition}
Let~$G$ be a graph, let~$T$ be an independent set in~$G$, and let~$\hat{G}_T$ be the graph obtained from~$G$ by eliminating~$T$; the order is irrelevant as~$T$ is independent. Then~$T$ is a \emph{treewidth-invariant set} if for every~$v \in T$, the graph~$\hat{G}_T$ is a minor of~$G - \{v\}$.
\end{definition}

\begin{lemma} \label{lemma:graph:reduction:step}
If~$T$ is a treewidth-invariant set in~$G$ and~$\Delta := \max _{v \in T} \deg_G(v)$, then~$\tw(G) = \max (\Delta, \tw(\hat{G}_T))$.
\end{lemma}
\begin{proof}
We prove that~$\tw(G)$ is at least, and at most, the claimed amount.

\textbf{($\geq$).} As~$\hat{G}_T$ is a minor of~$G$, we have~$\tw(G) \geq \tw(\hat{G}_T)$ (cf.~\cite{Bodlaender98}). If~$\tw(\hat{G}_T) \geq \Delta$ then this implies the inequality. So assume that~$\Delta > \tw(\hat{G}_T)$. Let~$v \in T$ have degree~$\Delta$. By assumption,~$\hat{G}_T$ is a minor of~$G - \{v\}$. It contains all vertices of~$N_G(v)$ since~$T$ is an independent set. As~$N_G(v)$ is a clique in~$\hat{G}_T$, there is a series of minor operations in~$G - \{v\}$ that turns~$N_G(v)$ into a clique. Performing these operations on~$G$ rather than~$G - \{v\}$ results in a clique on vertex set~$N_G[v]$ of size~$\deg_G(v) + 1 = \Delta + 1$: the set~$N_G(v)$ is turned into a clique, and~$v$ remains unchanged. Hence~$G$ has a clique with~$\Delta + 1$ vertices as a minor, which is known to imply (cf.~\cite{Bodlaender98}) that its treewidth is at least~$\Delta$.

\textbf{($\leq$).} Consider an optimal elimination order~$\hat{\pi}$ for~$\hat{G}_T$, which costs~$\tw(\hat{G}_T) + 1$ by \thmref{theorem:tw:eliminationcharacterization}. Form an elimination order~$\pi$ for~$G$ by first eliminating all vertices in~$T$ in arbitrary order, followed by the remaining vertices in the order dictated by~$\hat{\pi}$. Consider what happens when eliminating the graph~$G$ in the order given by~$\pi$. Each vertex~$v \in T$ that is eliminated incurs cost~$\deg_G(v) + 1 \leq \Delta + 1$: as~$T$ is an independent set, eliminations before~$v$ do not affect~$v$'s neighborhood. Once all vertices of~$T$ have been eliminated, the resulting graph is identical to~$\hat{G}_T$, by definition. As~$\pi$ matches~$\hat{\pi}$ on the vertices of~$V(G) \setminus T$, and~$\hat{\pi}$ has cost~$\tw(\hat{G}_T) + 1$, the total cost of elimination order~$\pi$ on~$G$ is~$\max (\Delta + 1, \tw(\hat{G}_T) + 1)$. By \thmref{theorem:tw:eliminationcharacterization} this completes this direction of the proof.
\qed
\end{proof}

\lemmaref{lemma:graph:reduction:step} shows that when a treewidth-invariant set is eliminated from a graph, its treewidth changes in a controlled manner. To exploit this insight in a kernelization algorithm, we have to find treewidth-invariant sets in polynomial time. While it seems difficult to detect such sets in all circumstances, we show that the $q$-expansion lemma can be used to find a treewidth-invariant set when the size of the graph is large compared to its vertex cover number. The following auxiliary graph is needed for this procedure.

\begin{definition}
Given a graph~$G$ with a vertex cover~$X \subseteq V(G)$, we define the bipartite \emph{non-edge connection graph~$H_{G,X}$}. Its partite sets are~$V(G) \setminus X$ and~$\binom{X}{2} \setminus E(G)$, with an edge between a vertex~$v \in V(G) \setminus X$ and a vertex~$x_{\{p,q\}}$ representing~$\{p,q\} \in \binom{X}{2} \setminus E(G)$ if~$v \in N_G(p) \cap N_G(q)$. 
\end{definition}

For disjoint vertex subsets~$S$ and~$T$ in a graph~$G$, we say that \emph{$S$ is saturated by $q$-stars into~$T$} if we can assign to every~$v \in S$ a subset~$f(v) \subseteq N_G(v) \cap T$ of size~$q$, such that for any pair of distinct vertices~$u,v \in S$ we have~$f(u) \cap f(v) = \emptyset$. Observe that an empty set can trivially be saturated by $q$-stars.

\begin{lemma} \label{lemma:expansion:is:invariant}
Let~$(G,X,k)$ be an instance of \TreewidthVC. If~$H_{G,X}$ contains a set~$T \subseteq V(G) \setminus X$ such that~$S := N_{H_{G,X}}(T)$ can be saturated by $2$-stars into~$T$, then~$T$ is a treewidth-invariant set.
\end{lemma}
\begin{proof}
As~$T$ is a subset of the independent set~$V(G) \setminus X$, the set~$T$ is independent in~$G$. It remains to prove that for every~$v \in T$, the graph~$\hat{G}_T$ is a minor of~$G - \{v\}$. So consider an arbitrary vertex~$v^* \in T$. We give a series of minor operations that transforms~$G - \{v^*\}$ into~$\hat{G}_T$. The crucial part of the transformation consists of contracting vertices of~$T \setminus \{v^*\}$ into vertices of~$X$, to turn~$N_G(v)$ into a clique for all~$v \in T$; afterwards we can simply delete all remaining vertices of~$T \setminus \{v^*\}$. Let~$f \colon S \to \binom{T}{2}$ be a mapping that assigns to each vertex in~$v$ a set of two of~$v$'s neighbors in~$T$, such that the images of~$f$ are pairwise disjoint.

Consider a vertex~$v \in T$ such that~$N_G(v)$ is not a clique. Let~$\{p,q\}$ be a non-edge in~$G[N_G(v)]$. As~$v$ is adjacent to both~$p$ and~$q$, vertex~$v$ is adjacent to the representative~$x_{\{p,q\}}$ in~$H_{G,X}$, implying that~$x_{\{p,q\}} \in S$. Hence~$x_{\{p,q\}}$ is saturated by a $2$-star into~$T$. Consider the two vertices~$f(x_{\{p,q\}})$ assigned to~$x_{\{p,q\}}$; at least one of them, say~$u$, differs from~$v^*$. As~$u$ is adjacent to~$x_{\{p,q\}}$ in~$H_{G,X}$ by definition of $2$-star saturation, by definition of~$H_{G,X}$ this implies that~$u$ is adjacent to both~$p$ and~$q$. Hence contracting~$u$ into~$p$ creates the missing edge~$\{p,q\}$. Now observe that as the images of~$f$ are pairwise disjoint, for each non-edge~$\{p,q\}$ in the neighborhood of some vertex in~$T$, there is a distinct vertex unequal to~$v^*$ that can be contracted to create the non-edge. Contracting all such vertices into appropriate neighbors therefore turns each set~$N_G(v)$ for~$v \in T$ into a clique. Hence we 
establish that~$\hat{G}_T$ is indeed a minor of~$G - \{v^*\}$, proving that~$T$ is treewidth-invariant.
\qed
\end{proof}

\begin{qexpansion}[{\cite[Lemma 12]{FominLMPS11}}]
Let~$q$ be a positive integer, and let~$m$ be the size of a maximum matching in a bipartite graph~$H$ with partite sets~$A$ and~$B$. If~$|B| > m \cdot q$ and there are no isolated vertices in~$B$, then there exist nonempty vertex sets~$S \subseteq A$ and~$T \subseteq B$ such that~$S$ is saturated by $q$-stars into~$T$ and~$S = N_H(T)$. Furthermore,~$S$ and~$T$ can be found in time polynomial in the size of~$H$ by a reduction to bipartite matching.
\end{qexpansion}

\begin{theorem}
\TreewidthVC has a kernel with~$\Oh(|X|^2)$ vertices that can be encoded in~$\Oh(|X|^3)$ bits.
\end{theorem}
\begin{proof}
Given an instance~$(G,X,k)$ of \TreewidthVC, the algorithm constructs the non-edge connection graph~$H_{G,X}$ with partite sets~$A = \binom{X}{2} \setminus E(G)$ and~$B = V(G) \setminus X$. We attempt to find a treewidth-invariant set~$T \subseteq B$. If~$B$ has an isolated vertex~$v$, then by definition of~$H_{G,X}$ the set~$N_G(v)$ is a clique implying that~$\{v\}$ is treewidth-invariant. If~$B$ has no isolated vertices, we apply the $q$-expansion lemma with~$q := 2$ to attempt to find a set~$S \subseteq A$ and~$T \subseteq B$ such that~$S$ is saturated by $2$-stars into~$T$ and~$S = N_{H_{G,X}}(T)$. Hence such a set~$T$ is  treewidth-invariant by \lemmaref{lemma:expansion:is:invariant}. If we find a treewidth-invariant set~$T$:
\begin{itemize}
	\item If~$\max _{v \in T} \deg_G(v) \geq k+1$ then we output a constant-size \no-instance, as \lemmaref{lemma:graph:reduction:step} then ensures that~$\tw(G) \geq \deg_G(v) > k$. 
	\item Otherwise we reduce to~$(\hat{G}_T, X, k)$ and restart the algorithm.
\end{itemize}

\noindent Each iteration takes polynomial time. As the vertex count decreases in each iteration, there are at most~$n$ iterations until we fail to find a treewidth-invariant set. When that happens, we output the resulting instance. The $q$-expansion lemma ensures that at that point,~$|B| \leq 2m$, where~$m$ is the size of a maximum matching in~$H_{G,X}$. As~$m$ cannot exceed the size of the partite set~$A$, which is bounded by~$\binom{|X|}{2}$ as there cannot be more non-edges in a set of size~$|X|$, we find that~$|B| \leq 2 \binom{|X|}{2}$ upon termination. As vertex set~$B$ of the graph~$H_{G,X}$ directly corresponds to~$V(G) \setminus X$, this implies that~$G$ has at most~$|X| + 2 \binom{|X|}{2}$ vertices after exhaustive reduction. Thus the instance that we output has~$\Oh(|X|^2)$ vertices. We can encode it in~$\Oh(|X|^3)$ bits: we store an adjacency matrix for~$G[X]$, and for each of the~$\Oh(|X|^2)$ vertices~$v$ in~$V(G) \setminus X$ we store a vector 
of~$|X|$ bits, indicating for each~$x \in X$ whether~$v$ is adjacent to it.
\qed
\end{proof}

\section{Conclusion}
In this paper we contributed to the knowledge of sparsification for \Treewidth by establishing lower and upper bounds. Our work raises a number of questions.

We showed that \Treewidth and \Pathwidth instances on~$n$ vertices are unlikely to be compressible into~$\Oh(n^{2 - \epsilon})$ bits. Are there natural problems on general graphs that do allow (generalized) kernels of size~$\Oh(n^{2 - \epsilon})$? Many problems admit $\Oh(k)$-vertex kernels when restricted to \emph{planar} graphs~\cite{BodlaenderFLPST09}, which can be encoded in~$\Oh(k)$ bits by employing succinct representations of planar graphs. Obtaining subquadratic-size compressions for NP-hard problems on classes of potentially \emph{dense} graphs, such as unit-disk graphs, is an interesting challenge.

In \sectref{section:tw:kernel} we gave a quadratic-vertex kernel for \TreewidthVC. While the algorithm is presented for the decision problem, it is easily adapted to the optimization setting (cf.~\cite{BodlaenderKE05}). The key insight for our reduction is the notion of treewidth-invariant sets, together with the use of the $q$-expansion lemma to find them when the complement of the vertex cover has superquadratic size. A challenge for future research is to identify treewidth-invariant sets that are not found by the $q$-expansion lemma; this might decrease the kernel size even further. As the sparsification lower bound proves that \TreewidthVC is unlikely to admit kernels of bitsize~$\Oh(|X|^{2 - \epsilon})$, while the current kernel can be encoded in~$\Oh(|X|^3)$ bits, an obvious open problem is to close the gap between the upper and the lower bound. Does \TreewidthVC have a kernel with~$\Oh(|X|)$ vertices? If not, then is there at least a kernel with~$\Oh(|X|^2)$ rather than~$\Oh(|X|^3$) edges?

For \PathwidthVC, a kernel with~$\Oh(|X|^3)$ vertices is known~\cite{BodlaenderJK12a}. Can this be improved to~$\Oh(|X|^2)$ using an approach similar to the one used here? The obvious pathwidth-analogue of \lemmaref{lemma:graph:reduction:step} fails, as removing a low-degree simplicial vertex may decrease the pathwidth of a graph. Finally, one may consider whether the ideas of the present paper can improve the kernel size for \Treewidth parameterized by a feedback vertex set~\cite{BodlaenderJK11b}.

% Note: Please add new references at the very BOTTOM of the file!
\bibliography{Paper}
\bibliographystyle{abbrvurl}

\end{document}